\documentclass[12pt]{amsart}

\usepackage{amssymb,amsthm,amsmath}
\usepackage[numbers,sort&compress]{natbib}
\usepackage{color}
\usepackage{graphicx}
\usepackage{tikz}
\usepackage{amssymb,amsthm,amsmath}
\usepackage[numbers,sort&compress]{natbib}
\usepackage{color}
\usepackage{graphicx}
\usepackage{tikz}
\usepackage{xparse}

\title[ A new proof of the sharp Gordon's lemma]{ A New Proof of the Sharp Gordon's Lemma:  No Eigenvalues for Schr\"odinger Operators with Almost Repetition Potentials}
\author{Wencai Liu}
\address[W. Liu]{ Department of Mathematics, Texas A\&M University, College Station, TX 77843-3368, USA} \email{liuwencai1226@gmail.com; wencail@tamu.edu}

\theoremstyle{plain}
\newtheorem{theorem}{Theorem}[section]

\newtheorem{corollary}[theorem]{Corollary}
\newtheorem{lemma}[theorem]{Lemma}

\newtheorem{remark}{Remark}
\newcommand{\C}{\mathbb{C}}

\newcommand{\Z}{\mathbb{Z}}

\theoremstyle{plain}

\newtheorem{conjecture}{Conjecture}

\begin{document}
	
	\maketitle
	\begin{center}
		{\it Dedicated to Barry Simon on the occasion of his 80th birthday}
	\end{center}
	\begin{abstract}
Building on  the work of Jitomirskaya-Simon  and Jitomirskaya-Liu, who established the absence of eigenvalues for Schr\"odinger operators with almost reflective repetition potentials, we provide a new proof of the sharp Gordon’s lemma, which asserts the absence of eigenvalues for Schr\"odinger operators with almost repetition potentials.
		
	\end{abstract}

	\section{Introduction}

In this paper, 	we study  one-dimensional discrete Schr\"odinger operators  given by:
		\begin{equation}\label{op}
		(Hu)(k) = u(k+1) + u(k-1) + V(k)u(k), \quad k \in \mathbb{Z},
		\end{equation}
		where \(V = \{V(k)\}\) is a bounded potential.

We focus on two types of potentials: almost reflective repetition potentials and almost repetition potentials.  These potentials are closely related to those studied in the context of quasi-periodic Schr\"odinger operators.
		
		A potential \(V = \{V(k)\}_{k \in \mathbb{Z}}\) is said to have {\(\gamma\)-repetitions} for some \(\gamma > 0\) (a modification of the definition in \cite{jz}) if there exists an increasing sequence \(k_n \to \infty\) or \(k_n \to -\infty\) such that, for all \(k \in \mathbb{Z}\):
		\begin{equation}\label{rep}
		|V(k_n + k) - V(k)| \leq e^{-\gamma |k_n|}.
		\end{equation}
		
		Similarly, a potential \(V = \{V(k)\}_{k \in \mathbb{Z}}\) is said to have {\(\gamma\)-reflective repetitions} if there exists a sequence \(k_n \to \infty\) or \(k_n \to -\infty\) such that, for all \(k \in \mathbb{Z}\):
		\begin{equation}\label{rep1}
		|V(k_n - k) - V(k)| \leq e^{-\gamma |k_n|}.
		\end{equation}
		
		Potentials with  repetitions are often referred to as \textit{Gordon-type potentials}, while reflective repetition potentials are also known as \textit{palindromic  potentials}.

We define one-step transfer matrices
$$
A_k(E):=\begin{pmatrix}
E-V(k)&-1\\
1&0
\end{pmatrix}
$$
and the multi-step transfer matrices by
$$
A_{k,j}(E):=\begin{cases}
A_{k-1}(E)A_{k-2}(E)\cdots A_j(E),&j< k;\\
A_{k,j}=I,&k=j;\\
A_{j,k}^{-1}(E),&j>k.
\end{cases}
$$

Define \begin{equation}\label{Lambda}\mathcal{L}(E) :=
\limsup_{k\to \infty}\sup_{ j\in\Z} \frac{\ln \| A_{k+j,j}(E)\|}{k}.\end{equation} 

It is clear that for bounded \(V\), we have \(\mathcal{L}(E) < \infty\) for every \(E\).

The foundational work of Gordon \cite{gor} and Simon \cite{sim1}, commonly referred to as the Gordon-type argument, laid the groundwork for proving that for $\gamma$-repetition potentials, $H$ has no eigenvalues in the regime

$$\{E \in \mathbb{C} : \gamma > 2\mathcal{L}(E)\}.$$

Although Gordon and Simon did not explicitly prove this exact result, it can be established by carefully analyzing their proofs.
Gordon first presented his argument in the 1970s in Russian \cite{gor}, but it remained largely unfamiliar to the broader mathematical community until Barry Simon provided a detailed discussion of the Gordon argument in his 1982   article \cite{sim1}  (also consult with \cite{as82}). For additional historical context, see \cite{Sasha}. 
	
	Recently,
	Avila-You-Zhou \cite{ayz} improved Gordon type arguments, extending the regime to
	\[
	\{E \in \mathbb{C} : \gamma > \mathcal{L}(E)\}.
	\]
	The absence of eigenvalues in the regime \(\{E \in \mathbb{C} : \gamma > \mathcal{L}(E)\}\) is sharp, as demonstrated by the example of almost Mathieu operators (see the definition later) \cite{ayz, jl1, jk}.
	
Over the past decades, Gordon-type arguments have been widely applied to the spectral theory of various models \cite{jl, yang, liu3, jko, simm, jk24, f17, dam20, dam04,liupmj,yzzjst,dfm02,fgpw}.
	
	For \(\gamma\)-reflective repetition potentials, Jitomirskaya and Simon \cite{js} proved that \(H\) has no eigenvalues in
	 the regime
	 \begin{equation}\label{gdec28}
\{E \in \mathbb{C} : \gamma > 2\mathcal{L}(E)\}.
	 \end{equation}
 
  Recently, Jitomirskaya and the author improved the threshold, reducing the numerical number in \eqref{gdec28} from \(2\) to \(1\) \cite{jl2}. This result is also sharp, as evidenced by the almost Mathieu operator \cite{jl2} (also  see Ge-You-Zhou \cite{gyz24}).
 The Jitomirskaya-Simon argument has been commonly used to study Schr\"odinger operators with reflective repetition potentials, e.g.,  \cite{dk, hks95,jlm,CF24}.
	
	In the present paper, we build on the approaches developed by Jitomirskaya-Simon \cite{js} and Jitomirskaya-Liu \cite{jl2} to study almost reflective repetition potentials. Our methods provide a new proof of the absence of eigenvalues for almost repetition potentials.

\begin{theorem}\label{mthm}
	Assume \(\gamma > \mathcal{L}(E)\). Then the eigen-equation \(Hu = Eu\) has no \(\ell^2(\mathbb{Z})\) solutions.
\end{theorem}

Consider the (discrete) quasi-periodic Schr\"{o}dinger operator on \(\ell^2(\mathbb{Z})\):
\begin{equation}\label{op1}
(H_{v,\alpha,\theta}u)(k) = u(k+1) + u(k-1) + v(\theta + k\alpha)u(k), \quad k \in \mathbb{Z},
\end{equation}
where \(v: \mathbb{T} = \mathbb{R}/\mathbb{Z} \to \mathbb{R}\) is a continuous potential, \(\alpha\) is the frequency, and \(\theta\) is the phase.

Let \({L}(E)\) denote the Lyapunov exponent of the operator \eqref{op1}:
\begin{equation}
{L}(E) = \limsup_{j \to \infty} \frac{1}{j} \int_{\mathbb{T}} \ln \|T_{j,0}(E)\| \, d\theta = \lim_{j \to \infty} \frac{1}{j} \int_{\mathbb{T}} \ln \|T_{j,0}(E)\| \, d\theta.
\end{equation}
For quasi-periodic Schr\"{o}dinger operators with continuous potentials, it is known that \(L(E) = \mathcal{L}(E)\) (independent of $\theta$).

Define 
\[
\beta(\alpha) = \limsup_{k \to \infty} \frac{-\ln \|k\alpha\|_{\mathbb{T}}}{k},
\]
where \(\|x\|_{\mathbb{T}} = \mathrm{dist}(x, \mathbb{Z})\).

\begin{corollary}\label{cor1}
	Assume \(v\) is \(\kappa\)-H\"older continuous. Then \(H\) has no eigenvalues in  the regime \(\{E \in \C: {L}(E) < \kappa \beta(\alpha)\}\).
\end{corollary}

By setting \(v(\theta) = 2\lambda \cos(2\pi \theta)\) in \eqref{op1}, we obtain the almost Mathieu operator \(H_{\lambda,\alpha,\theta}\). For the almost Mathieu operator, the Lyapunov exponent on the spectrum is given by \({L}(E) = \max\{\ln |\lambda|, 0\}\) \cite{bj}.

As a consequence, Corollary \ref{cor1} implies:

\begin{corollary}\label{cor2}
	Assume \(\alpha\) satisfies \(\beta(\alpha) > 0\). For the almost Mathieu operator, \(H_{\lambda,\alpha,\theta}\) has no eigenvalues when \(|\lambda| < e^{\beta(\alpha)}\).
\end{corollary}

\begin{remark}
	As mentioned earlier, Corollary \ref{cor2} was previously proved using an improved Gordon-type argument \cite{ayz}.
\end{remark}
	\section{Proof of Theorem \ref{mthm}}
	Without loss of generality, assume that $k_n\to \infty$. The proof of  the case $k_n\to -\infty$ is similar.
	Since  the spectrum of $H=\Delta+V$ is  bounded, we can assume $E$ is bounded. 
	By the definition of $\mathcal L(E)$, one has that
	\begin{equation}\label{lya1}
	||	T_{m,j}(E)||\leq Ce^{(\mathcal L (E)+\varepsilon)|m-j|},
	\end{equation}
	where the constant $C$ depends on $\varepsilon$ and the potential $V$. 
	\begin{lemma}\label{lem1}
		For all $m,j\in\Z$,  we have
		\begin{equation*}
		|| T_{m,j} (E)-T_{m+k_n,j+k_n}(E)||\leq C e^{-\gamma  k_n}  e^{(\mathcal L(E)+\varepsilon) |m-j|}. 
		\end{equation*}

	\end{lemma}
	\begin{proof}
		By \eqref{lya1}, 
		the proof  of Lemma \ref{lem1}  follows from the standard telescoping. See \cite[Lemma  4.5]{simm} for example.
	\end{proof}
	\begin{proof}[\bf Proof of Theorem \ref{mthm} ]
		We prove it  by contradiction.   Suppose there exists $E$	with $\mathcal L(E)< \gamma$   such that the eigen-equation $Hu=Eu$ has an $\ell^2(\Z)  $ solution  $u$.  
	Assume $u$ is normalized, namely
		\begin{equation*}
			|| u||_{\ell^2} = \sum_{k\in\Z}|u(k)|^2=1.
		\end{equation*}
		For each $n\in\Z_+$, define  sequences $\{u_n(k)\}_{k\in \Z}$ and $\{V_n(k)\}_{k\in \Z}$
		as   $u_n(k)= u(k_n+k)$ and  $V_n(k)=V(k_n+k)$, $k\in\Z$.
		Then by \eqref{rep},  one has that
		\begin{equation}\label{Eqp}
			|V(k)-V_n(k)|\leq e^{- \gamma k_n },\text{ for all } k\in\Z.
		\end{equation}
		We also have
		\begin{equation}\label{Equ}
			u(k+1)+u(k-1)+ V(k)u(k)=Eu(k)
		\end{equation}
		and
		\begin{equation}\label{Equi}
			u_n(k+1)+u_n(k-1)+ V_n(k)u_n(k)=Eu_n(k).
		\end{equation}
		Let $W(k)=W(f,g)=f(k+1)g(k)-f(k)g(k+1)$ be the
		Wronskian.

		By a standard calculation using (\ref{Eqp}), (\ref{Equ}) and (\ref{Equi}),
		we have
		\begin{eqnarray}
			|W(u,u_n)(k)-W(u,u_n)(k-1)| &\leq & |V(k)-V_n(k)||u(k)u_n(k)|  \nonumber \\
			&\leq &   e^{-\gamma k_n}|u(k)u_n(k)|. \nonumber
		\end{eqnarray}
		This implies for any $m\geq 0$ and $k$,
		\begin{eqnarray}
			|W(u,u_n)(k+m)-W(u,u_n)(k-1)| &\leq &  e^{-\gamma k_n}\sum_{j=0}^{m} |u(k+j)u_n(k+j)|  \nonumber \\
			&\leq &   e^{-\gamma k_n}, \label{Equ7}
		\end{eqnarray}
		where the second inequality holds by the  Cauchy–Schwarz inequality and 
		the fact that $||u||_{\ell^2}=||u_n||_{\ell^2}=1.$
		
		Notice that 
		$\sum_k|W(u,u_n)(k)|\leq 2$. Therefore, 
		\begin{equation}\label{g6}
		\lim_{k\to \infty}	|W(u,u_n)(k)|=0.
		\end{equation}
		By (\ref{Equ7}) and \eqref{g6}, we must have that for all $k\in\Z$,
		\begin{equation}\label{Equ8}
			|W(u,u_n)(k)|\leq  e^{-\gamma k_n}.
		\end{equation}

		Letting $k=-1$ in \eqref{Equ8}, one has that 
		\begin{equation}\label{g1}
			|W(u,u_n)(-1)|=
			\left |	\det \begin{pmatrix}
				u(0) & u(k_n)\\ 
				u(-1) & u(k_n-1)
			\end{pmatrix}\right |\leq e^{-\gamma k_n}.
		\end{equation}

		Write $\begin{pmatrix}
			u(k_n)\\ 
			u(k_n-1)
		\end{pmatrix}$  as the linear combination of $\begin{pmatrix}
			u(0) \\ 
			u(-1) 
		\end{pmatrix}$  and $ \begin{pmatrix}
			-u(-1)  \\ 
			u(0)  
		\end{pmatrix}$, namely
		\begin{equation}\label{g2}
			\begin{pmatrix}
				u(k_n)\\ 
				u(k_n-1)
			\end{pmatrix}=a_n\begin{pmatrix}
				u(0) \\ 
				u(-1) 
			\end{pmatrix}+b_n \begin{pmatrix}
				-u(-1)  \\ 
				u(0)  
			\end{pmatrix}.
		\end{equation}
		Since $u$ is non-trivial, one has that
		\begin{equation}\label{g100}
		|u(0)|^2+|u(-1)|^2>0.
		\end{equation}
		By \eqref{g1}, \eqref{g2} and \eqref{g100}, we have that 
		\begin{equation}\label{g3}
			b_n=O(e^{-\gamma k_n}).
		\end{equation}
		Since $u\in\ell^2(\Z)$, one has that
		\begin{equation}\label{g4}
			a_n=o(1),
		\end{equation}
		as $n\to\infty$. 
By the definition of transfer matrices, one has that
		\begin{equation}\label{g101}
			\begin{pmatrix}
				u(k) \\ 
				u(k-1) 
			\end{pmatrix} =T_{k,m} (E) \begin{pmatrix}
				u(m) \\ 
				u(m-1) 
			\end{pmatrix}. 
		\end{equation}

		Multiplying $T_{0,k_n}$ on both sides  of \eqref{g2},  we conclude  that 
		\begin{eqnarray}
		\begin{pmatrix}
		u(0)\\ 
		u(-1)
		\end{pmatrix}&=&a_nT_{0,k_n}(E)\begin{pmatrix}
		u(0) \\ 
		u(-1) 
		\end{pmatrix}+b_n T_{0,k_n}(E) \begin{pmatrix}
		-u(-1)  \\ 
		u(0)  
		\end{pmatrix}  \label{g103}\\
			 &=&a_nT_{0,k_n}(E)\begin{pmatrix}
				u(0) \\ 
				u(-1) 
			\end{pmatrix}+ O(e^{-(\gamma -\mathcal L(E)-\varepsilon)k_n})  \label{g104}\\
			&=&a_n T_{-k_n,0}(E)\begin{pmatrix}
				u(0) \\ 
				u(-1) 
			\end{pmatrix}+ O(e^{-(\gamma -\mathcal L(E)-\varepsilon)k_n}) \label{g105}\\
			&=&a_n \begin{pmatrix}
				u(-k_n) \\ 
				u(-k_n-1) 
			\end{pmatrix}+  O(e^{-(\gamma -\mathcal L(E)-\varepsilon)k_n}), \label{g106}
		\end{eqnarray}
		where \eqref{g103}  and \eqref{g106} hold by  \eqref{g101},   \eqref{g104} holds by \eqref{lya1} and \eqref{g3},
		and \eqref{g105} holds by
		Lemma \ref{lem1}.
		The equality  \eqref{g106}  can  not hold   since  $ \left \| \begin{pmatrix}
			u(-k_n) \\ 
			u(-k_n-1) 
		\end{pmatrix}\right \|=o(1)$ and $a_n=o(1)$ as $n\to \infty$. We finish the proof.
	\end{proof}
	
	\section*{Acknowledgments}
	
	W. Liu was a 2024-2025 Simons fellow.  
	W. Liu   was supported by NSF DMS-2246031 and DMS-2052572.


\begin{thebibliography}{10}
	
	\bibitem{ayz}
	A.~Avila, J.~You, and Q.~Zhou.
	\newblock Sharp phase transitions for the almost {M}athieu operator.
	\newblock {\em Duke Math. J.}, 166(14):2697--2718, 2017.
	
	\bibitem{as82}
	J.~Avron and B.~Simon.
	\newblock Singular continuous spectrum for a class of almost periodic {J}acobi
	matrices.
	\newblock {\em Bull. Amer. Math. Soc. (N.S.)}, 6(1):81--85, 1982.
	
	\bibitem{bj}
	J.~Bourgain and S.~Jitomirskaya.
	\newblock Continuity of the {L}yapunov exponent for quasiperiodic operators
	with analytic potential.
	\newblock volume 108, pages 1203--1218. 2002.
	\newblock Dedicated to David Ruelle and Yasha Sinai on the occasion of their
	65th birthdays.
	
	\bibitem{CF24}
	C.~Cedzich and J.~Fillman.
	\newblock Absence of bound states for quantum walks and {CMV} matrices via
	reflections.
	\newblock {\em J. Spectr. Theory}, 14(4):1513--1536, 2024.
	
	\bibitem{dam20}
	D.~Damanik.
	\newblock Gordon-type arguments in the spectral theory of one-dimensional
	quasicrystals.
	\newblock In {\em Directions in mathematical quasicrystals}, volume~13 of {\em
		CRM Monogr. Ser.}, pages 277--305. Amer. Math. Soc., Providence, RI, 2000.
	
	\bibitem{dfm02}
	D.~Damanik.
	\newblock Absence of eigenvalues for a class of {S}chr\"{o}dinger operators on
	the strip.
	\newblock {\em Forum Math.}, 14(5):797--806, 2002.
	
	\bibitem{dam04}
	D.~Damanik.
	\newblock A version of {G}ordon's theorem for multi-dimensional
	{S}chr\"{o}dinger operators.
	\newblock {\em Trans. Amer. Math. Soc.}, 356(2):495--507, 2004.
	
	\bibitem{dk}
	D.~Damanik and R.~Killip.
	\newblock Reflection symmetries of almost periodic functions.
	\newblock {\em J. Funct. Anal.}, 178(2):251--257, 2000.
	
	\bibitem{fgpw}
	L.~Fang, S.~Guo, Y.~Peng, and F.~Wang.
	\newblock Quantitative version of {G}ordon's lemma for {H}amiltonian with
	finite range.
	\newblock {\em Linear Algebra Appl.}, 687:91--107, 2024.
	
	\bibitem{f17}
	J.~Fillman.
	\newblock Purely singular continuous spectrum for {S}turmian {CMV} matrices via
	strengthened {G}ordon lemmas.
	\newblock {\em Proc. Amer. Math. Soc.}, 145(1):225--239, 2017.
	
	\bibitem{gyz24}
	L.~Ge, J.~You, and Q.~Zhou.
	\newblock Structured quantitative almost reducibility and its applications.
	\newblock {\em arXiv preprint arXiv:2407.05490}, 2024.
	
	\bibitem{gor}
	A.~Y. Gordon.
	\newblock The point spectrum of the one-dimensional {S}chr\"odinger operator.
	\newblock {\em Uspekhi Matematicheskikh Nauk}, 31(4):257--258, 1976.
	
	\bibitem{hks95}
	A.~Hof, O.~Knill, and B.~Simon.
	\newblock Singular continuous spectrum for palindromic {S}chr\"{o}dinger
	operators.
	\newblock {\em Comm. Math. Phys.}, 174(1):149--159, 1995.
	
	\bibitem{jk}
	S.~Jitomirskaya and I.~Kachkovskiy.
	\newblock {$L^2$}-reducibility and localization for quasiperiodic operators.
	\newblock {\em Math. Res. Lett.}, 23(2):431--444, 2016.
	
	\bibitem{jk24}
	S.~Jitomirskaya and I.~Kachkovskiy.
	\newblock Sharp arithmetic delocalization for quasiperiodic operators with
	potentials of semi-bounded variation.
	\newblock {\em arXiv preprint arXiv:2408.16935}, 2024.
	
	\bibitem{jko}
	S.~Jitomirskaya and S.~Koci\'{c}.
	\newblock Spectral theory of {S}chr\"{o}dinger operators over circle
	diffeomorphisms.
	\newblock {\em Int. Math. Res. Not. IMRN}, (13):9810--9829, 2022.
	
	\bibitem{jl}
	S.~Jitomirskaya and W.~Liu.
	\newblock Arithmetic spectral transitions for the {M}aryland model.
	\newblock {\em Comm. Pure Appl. Math.}, 70(6):1025--1051, 2017.
	
	\bibitem{jl1}
	S.~Jitomirskaya and W.~Liu.
	\newblock Universal hierarchical structure of quasiperiodic eigenfunctions.
	\newblock {\em Ann. of Math. (2)}, 187(3):721--776, 2018.
	
	\bibitem{jl2}
	S.~Jitomirskaya and W.~Liu.
	\newblock Universal reflective-hierarchical structure of quasiperiodic
	eigenfunctions and sharp spectral transition in phase.
	\newblock {\em J. Eur. Math. Soc. (JEMS)}, 26(8):2797--2836, 2024.
	
	\bibitem{jlm}
	S.~Jitomirskaya, W.~Liu, and L.~Mi.
	\newblock Sharp palindromic criterion for semi-uniform dynamical localization.
	\newblock {\em arXiv preprint arXiv:2410.21700}, 2024.
	
	\bibitem{Sasha}
	S.~Jitomirskaya, S.~A. Molchanov, B.~Simon, and B.~R. Vainberg.
	\newblock Alexander {G}ordon.
	\newblock {\em J. Spectr. Theory}, 9(4):1157--1164, 2019.
	
	\bibitem{js}
	S.~Jitomirskaya and B.~Simon.
	\newblock Operators with singular continuous spectrum. {III}. {A}lmost periodic
	{S}chr\"{o}dinger operators.
	\newblock {\em Comm. Math. Phys.}, 165(1):201--205, 1994.
	
	\bibitem{yang}
	S.~Jitomirskaya and F.~Yang.
	\newblock Singular continuous spectrum for singular potentials.
	\newblock {\em Comm. Math. Phys.}, 351(3):1127--1135, 2017.
	
	\bibitem{jz}
	S.~Jitomirskaya and S.~Zhang.
	\newblock Quantitative continuity of singular continuous spectral measures and
	arithmetic criteria for quasiperiodic {S}chr\"{o}dinger operators.
	\newblock {\em J. Eur. Math. Soc. (JEMS)}, 24(5):1723--1767, 2022.
	
	\bibitem{liupmj}
	W.~Liu.
	\newblock Small denominators and large numerators of quasiperiodic
	{S}chr{\"o}dinger operators.
	\newblock {\em Peking Mathematical Journal to appear}.
	
	\bibitem{liu3}
	W.~Liu.
	\newblock Continuous quasiperiodic {S}chr\"{o}dinger operators with {G}ordon
	type potentials.
	\newblock {\em J. Math. Phys.}, 59(6):063501, 6, 2018.
	
	\bibitem{sim1}
	B.~Simon.
	\newblock Almost periodic {S}chr{\"o}dinger operators: a review.
	\newblock {\em Advances in Applied Mathematics}, 3(4):463--490, 1982.
	
	\bibitem{simm}
	B.~Simon.
	\newblock Almost periodic {S}chr\"{o}dinger operators. {IV}. {T}he {M}aryland
	model.
	\newblock {\em Ann. Physics}, 159(1):157--183, 1985.
	
	\bibitem{yzzjst}
	J.~You, S.~Zhang, and Q.~Zhou.
	\newblock Point spectrum for quasi-periodic long range operators.
	\newblock {\em J. Spectr. Theory}, 4(4):769--781, 2014.
	
\end{thebibliography}
\end{document}